\documentclass[11pt,letterpaper]{article}

\usepackage[margin=1in]{geometry}
\usepackage{lmodern}
\usepackage[T1]{fontenc}
\usepackage[utf8]{inputenc}
\usepackage[tbtags]{amsmath}
\allowdisplaybreaks
\usepackage{amssymb,amsthm,mathtools, bbm}
\usepackage{thm-restate}
\usepackage{physics}
\usepackage{hyperref}
\usepackage[svgnames]{xcolor}
\usepackage[capitalise,nameinlink]{cleveref}
\definecolor{links}{RGB}{11, 85, 255}
\definecolor{cites}{RGB}{0, 200, 0}
\definecolor{urls}{RGB}{255, 116, 0}
\hypersetup{colorlinks={true},linkcolor={links},citecolor=[named]{cites},urlcolor={urls}}
\usepackage[numbers,compress]{natbib}
\usepackage{doi}
\usepackage{tikz} 
\usepackage{pgfplots}
\pgfplotsset{compat=1.14}
\usepgfplotslibrary{fillbetween}
\usepackage{hyphenat}
\usepackage[font=footnotesize]{caption}
\usepackage{subcaption}
\usepackage{appendix}
\usepackage{enumitem}

\usepackage{todonotes}
\usepackage{nicefrac}

\usepackage[ruled]{algorithm2e}
\SetKwComment{Comment}{/* }{ */}
\SetKwFunction{cutter}{Cutter}
\SetKwFunction{chooser}{Chooser}
\SetKwFunction{pricer}{Pricer}
\SetKwProg{Fn}{}{:}{}

\usepackage{amsmath}

\newcommand{\cM}{\mathcal{M}}

\newcommand{\cT}{\mathcal{T}}

\newcommand{\R}{\mathbb{R}}
\newcommand{\F}{\mathbb{F}}

\newcommand{\ind}[1]{\mathbbm{1}\left\{#1\right\}}

\newcommand{\items}{M}
\newcommand{\agents}{N}

\newcommand{\signals}[1]{S_{#1}}
\newcommand{\en}[1]{\alpha_{#1}}
\newcommand{\allsignals}{S}

\newcommand{\prices}{P}
\newcommand{\mech}{\cM}
\newcommand{\allocations}{\mathcal{A}}
\newcommand{\prop}[3]{\text{PROP}_{#1}(#2,#3)}
\newcommand{\mms}[3]{\text{MMS}_{#1}(#2,#3)}
\newcommand{\aps}[3]{\text{APS}_{#1}(#2,#3)}
\newcommand{\inst}{\mathcal{I}}

\newtheorem{theorem}{Theorem}[section]

\newtheorem{corollary}[theorem]{Corollary}

\newtheorem{remark}[theorem]{Remark}
\newtheorem{observation}[theorem]{Observation}

\theoremstyle{definition}
\newtheorem{definition}[theorem]{Definition}

\newtheorem{example}[theorem]{Example}

\begin{document}

\title{\huge Fair Division with Interdependent Values}
\date{}
\author{Georgios Birmpas\thanks{Sapienza University of Rome, Email:  \texttt{birbas@diag.uniroma1.it}} \and Tomer Ezra\thanks{Sapienza University of Rome, Email:  \texttt{ezra@diag.uniroma1.it}} \and Stefano Leonardi\thanks{Sapienza University of Rome, Email:  \texttt{leonardi@diag.uniroma1.it}} \and Matteo Russo\thanks{Sapienza University of Rome, Email:  \texttt{mrusso@diag.uniroma1.it}}}

\maketitle

\begin{abstract}
    We introduce the study of designing allocation mechanisms for fairly allocating indivisible goods in settings with interdependent valuation functions.
    In our setting, there is a set of goods that needs to be allocated to a set of agents (without disposal).
    Each agent is given a private signal, and his valuation function depends on the signals of all agents.    
    Without the use of payments, there are strong impossibility results for designing strategyproof allocation mechanisms even in settings without interdependent values.
    Therefore, we turn to design mechanisms that always admit equilibria that are fair with respect to their true signals, despite their potentially distorted perception.
    To do so, we first extend the definitions of pure Nash equilibrium and well-studied fairness notions in literature 
    to the interdependent setting.
    We devise simple allocation mechanisms that always admit a fair equilibrium with respect to the true signals. We complement this result by showing that, even for very simple cases with binary additive interdependent valuation functions, no allocation mechanism that always admits an equilibrium, can guarantee that all equilibria are fair with respect to the true signals.
    \end{abstract}

\section{Introduction}\label{sec:intro}

The problem of fair division centers around the challenge of allocating a collection of resources to a set of individuals in a fair way. The roots of this problem can be traced back to the early work of Banach, Knaster, and Steinhaus \citep{Steinhaus49}, who introduced the notion of \emph{proportionality}, a concept that demands each person receives at least an equal share of the total value. Another prominent concept in the realm of fairness is \emph{envy-freeness} \citet{Foley67,GamovS58,Varian74}, which dictates that each individual values their resources as much as anyone else's. However, when resources are indivisible, proportionality and envy-freeness become impossible to attain in general. For example, consider an instance with two agents and one good, being positively valued by both: No allocation exists where no agent envies the other, or every agent gets his proportional share. This motivated the introduction of new (relaxed) fairness notions, such as EF1 (Envy-Freeness up to One Good) \citep{LiptonMMS04,Budish11}, EFX (Envy-Freeness up to Any Good) \citep{CaragiannisKMPS19}, and MMS (Maximin Share) \citep{Budish11}, to grapple with the division of indivisible resources to agents with equal entitlements, and APS (the AnyPrice Share) \citep{BabaioffEF21}, WMMS (Weighted Maximin Share) \citep{FarhadiGHLPSSY19}, WEF (Weighted Envy-Freeness) \citep{ChakrabortyISZ21} or $\ell$-out-of-$d$ share \citep{BabaioffNT21}, to agents with unequal entitlements. For an overview of results on the area, we refer the reader to the survey of \citet{AABFRLMVW22}.

When combining the problem of fairly allocating items with incentives, the problem becomes much more challenging. \citet{CKKK09}, \citet{MarkakisP11} and \citet{ABM16} introduced the strategic version of the problem, where the agents are assumed to be selfish, and their goal is to maximize their own utility. In particular, they considered the question of whether it is possible to have truthful allocation mechanisms (without payments) that provide fairness guarantees. This question was later resolved by \citep{AmanatidisBCM17}, who showed that truthfulness and fairness are incompatible even for the case of additive valuation functions, as no meaningful fairness notion can be guaranteed by a truthful allocation mechanism. This impossibility, led subsequent works to pursue positive results regarding truthfulness and fairness in more specialized valuation function settings, such as dichotomous additive \citep{Aziz20,HalpernPPS20}, matroid rank functions \citep{BabaioffEzFe21,BarmanV22,ViswanathanZ22}, or combining fairness with weaker versions of truthfulness, e.g., \citep{PV22}. In a work that is mostly related to ours, \citet{AmanatidisBFLLR21} followed a different approach and explored the fairness properties of the pure Nash equilibria (PNE) of non-truthful mechanisms. 
They focused on the case of additive agents, and showed that there are mechanisms that always admit PNE, all of which induce fair allocations according to the agents' true valuation functions. \citet{AmanatidisBLLR23} later expanded the applicability of these results to richer valuation function classes (e.g. cancellable, and submodular).

All the mechanisms that have been devised in the context of strategic fair division assume the agent's valuation functions to be \emph{independent} of each other. In several scenarios, however, such an assumption is either too strong or unrealistic. The concept of agents with interdependent valuation functions was first introduced in the context of auctions by \citet{Wilson69} and \citet{MilgromW82}. It has recently gained a lot of attention in the algorithmic game theory community, for efficient and strategyproof implementation of IDV auctions \citep{JehielM01, EdenFFG18, EdenFFGK19, EdenFTZ21, EdenGZ22, LuSZ22} and IDV public projects \citep{CohenFMT23} (we defer the interested reader to \Cref{sec:related}). However, its applicability goes further beyond that. In the case of fair division, consider, for example, a couple of heirs that need to partition (indivisible) inheritance goods between them: Each of them possesses some information (signal) about the goods, but their value depends on the other's signal too. The strategic nature of the situation is apparent, as it becomes clear that both parties have a vested interest in distorting the truth regarding their signals to manipulate the other party's perception and ultimately gain an advantage in the allocation of goods. Whether by exaggerating or downplaying the value of certain goods, each party hopes to influence the other's perception and come out ahead in the final allocation. Incorporating private signals into the mechanism design process poses a critical challenge now that the goal becomes to ensure fairness that aligns with the true valuation functions of all agents involved, as opposed to the perceived ones. 
Since allocation mechanisms that do not use monetary transfers cannot compensate the agents for revealing their private signals, having guarantees with respect to the true signal is even more challenging. 
Our main goal is to design allocation mechanisms that seek to establish equilibria, where the allocation is fair with respect to the agents' \emph{true} values, despite their potentially distorted perception.

\subsection{Our Contributions}
Before we begin, we would like to highlight that although the aforementioned impossibility results regarding truthfulness and fairness in the independent value model transfer to the setting of interdependent values, the positive results that regard the fair PNE of non-truthful mechanisms do not. Therefore, we initiate the study of mechanisms that admit the existence of fair equilibria with respect to the true signals of the agents, in settings where their values are interdependent. 

We extend several notions of fairness, along with the notion of pure Nash equilibrium, to the interdependent value model, and we design a class of mechanisms that, for every signal vector, have at least one pure Nash equilibrium with the following property: the allocation that corresponds to it, is fair with respect to the \emph{true} values (in contrast to what has been reported or perceived by the agents). The common characteristic between the mechanisms of this class is that each agent is not only required to report her own signal, but is also asked to report a \emph{guess} for everybody else's signal. We refer to the mechanisms of this class, as \emph{guessing} mechanisms. 
We now demonstrate, through an example, that \emph{guessing} mechanisms are a necessity in the IDV model as if an agent is not allowed to submit an additional bid beyond her own signal, then there might be cases where no pure Nash equilibrium is fair, according to any meaningful fairness notion.
\begin{example}\label{example}
    Let us consider an instance with two agents and the set of four goods $\items = \{a, b, c, d\}$, along with a mechanism $\cM$ that always has at least one fair equilibrium for every signal. Agent 1 values good $j \in M$ just according to her signal, i.e., $v_{1j} = s_{1j}$. Agent 2 instead values the same good according to agent 1's signal, i.e., $v_{2j} = s_{1j}$. Without loss of generality, assume that agent 2 is not able to report additional information beyond her own signal. Even more restrictively, his declaration cannot affect the mechanism's output at all. Consider the following instance where the values agent 1 (that also defines agent 2's) are $v_1 = (1,1,1,1) = v_2$. By our assumption on $\cM$, we know that there exists a report $r_1$ that agent 1 can declare, which induces a fair pure Nash equilibrium according to the true value. This corresponds to allocation $(A_1, A_2)$, where $A_1, A_2 \neq \emptyset$ for the allocation to be fair. The specifics of this allocation depend on the fairness notion that we examine. We now consider a new valuation function for agent 1 that associates $v_{1j} = s_{1j} = 1$ if $j \in A_1$, and $0$ otherwise. For this new instance, in any pure Nash equilibrium, agent 1 should receive all goods in $A_1$, as otherwise she could declare $r_1$ and get them. This means that the goods that agent 2 receives have $0$ value for him in any pure Nash equilibrium. Thereby, this shows the absence of fairness with respect to any meaningful fairness notion, a contradiction.
\end{example}

\noindent Below we present a roadmap to our paper:
\begin{itemize}[leftmargin=4mm]
    \item In \Cref{sec:equal}, we consider the case of two \textit{general} monotone agents with \textit{equal} entitlements, and we design a mechanism that is based on the \emph{cut and choose protocol}. Our mechanism  (Mechanism \ref{alg:cutchoose}), induces an equilibrium where the corresponding allocation is MMS and EFX for the cutter and EF for the chooser.
    \item In \Cref{sec:unequal}, we consider the case of two \textit{general} monotone agents with \textit{unequal} entitlements, and we devise a mechanism (Mechanism \ref{alg:pricechoose}) that induces an equilibrium where the corresponding allocation gives at least the APS to one agent, and the proportional share to the other agent as long as her valuation function is XOS. We, in fact, show that achieving the same fairness guarantees is impossible if the valuation functions of both agents are subadditive. As a byproduct, we emphasize that Mechanism \ref{alg:pricechoose} for the special case of independent additive values guarantees that both agents receive their APS in every equilibrium, which to our knowledge, is the first fairness guarantee in a strategic setting for agents with arbitrary entitlement and additive valuations\footnote{\citet{SuksompongT22} and \citet{SuksompongT23} showed fairness guarantees in a strategic setting for agents with unequal entitlements in the special cases of binary valuation functions and matroid rank valuation functions respectively.}.
    \item In \Cref{sec:multi}, we consider the case of three or more agents. There, we present a way of transforming, in a black-box manner, any algorithm into a mechanism with at least one pure Nash equilibrium that guarantees the same properties of the algorithm with respect to the true signals of the agent. 
    \item In \Cref{sec:main_imp}, we complement the above picture with the following negative result: it is impossible to construct allocation mechanisms that always admit a PNE, for which all equilibria are fair with respect to the true signals for either of the fairness notions of MMS and EF1. This impossibility is shown under agents with additive valuation functions and binary signals, and creates a stark separation with the independent valuation model considered in \citep{AmanatidisBFLLR21}.
\end{itemize}

\subsection{Further Related Work on Interdependent Mechanism Design with Money} \label{sec:related}

As previously underlined, mechanism design (with or without money) has traditionally been conceived for settings where agents' valuations are assumed to be independent. In this setting, a celebrated result by \citet{vickrey1961, Clarke1971, Groves1973} (the VCG mechanism) resolves the problem of truthful mechanism design optimally. However, the independence assumption falls short of characterizing natural scenarios, such as the illustrative oil drilling auction example of \citet{Wilson69}.  The auctioned land's value depends on how much oil there is underneath, and different agents might have different pieces of information (signals) about it. Most importantly, unlike traditional settings, each agent's value for that parcel of land depends on \emph{all} such pieces of information. In the context of single-item auctions, \citet{MilgromW82} formulated the \emph{interdependent value} (IDV) model, which prescribes each agent $i$ has a \emph{private} signal $s_i$ that summarizes the information she possesses about the item being auctioned. Furthermore, the value agent $i$ attributes to the item is a function of all signals, i.e., $v_i(s_i, s_{-i})$.

In the IDV model, it is impossible to design dominant strategy incentive-compatible mechanisms for single-item auctions in that if one agent does not report her signal truthfully, then another agent might not even be aware of its true valuation. More strongly, it is even impossible\footnote{\citet{JehielM01} show this impossibility even for Bayesian settings.} to design ex-post incentive compatible mechanisms \citep{JehielM01}, unless a condition called \emph{single-crossing} is met \citep{RoughgardenT13}. This signifies that each agent values her own signal above everybody else's. Beyond not being realistic, the single-crossing condition does not generalize well to multi-dimensional signal settings: \citet{EdenFFG18, EdenFTZ21, EdenGZ22}, and especially \citet{EdenFFGK19}, circumvent this issue (and later refined by \citet{LuSZ22}), extending ex-post incentive compatibility results to settings where valuations satisfy \emph{Submodularity over Signals} (SOS), a natural property to impose on valuations. Loosely, this means that the marginal increase in one's value due to her signal's increase is smaller the higher everybody else's signal is. Beyond auctions, \citet{CohenFMT23} have recently studied the problem of interdependent public projects.

\section{Model}


We study the problem of fairly allocating a set $\items$ of $m$ indivisible goods to a set $\agents=[n]$ of $n\geq 2$ agents with interdependent valuations. Each agent $i\in\agents$ is characterized by the following three parameters: (1) a private signal $s_i$ from the set of potential signals $\signals{i}\neq \emptyset$, where we denote by $\allsignals$ the Cartesian product $\bigtimes_{i\in \agents} \signals{i} $; (2) a publicly known entitlement $\en{i} \in (0,1)$; and (3) a publicly\footnote{The assumption that the valuations are public (and only the signals are private) is without loss of generality since one can use some dedicated bits of the private signal to encode the private valuation, and these bits will not influence the other agents' valuations.} known monotone and normalized\footnote{A valuation function $v$ is monotone if for every signal vector $s$ and sets $T\subseteq T' \subseteq \items$ it holds that $v(s,T)\leq v(s,T')$.  A valuation function $v$ is normalized if for every signal vector $s$ it holds $v(s,\emptyset) =0$.} combinatorial valuation $v_i:\allsignals \times 2^\items \rightarrow \R_{\geq 0}$.
We assume without loss of generality that the sum of entitlements is $1$, and we denote by $\alpha$ the vector of entitlements $(\en{1},\ldots,\en{n})$.
A special case is when all entitlements are $\nicefrac{1}{n}$, (i.e., equal entitlements).
An allocation $A=(A_1,\ldots,A_n)$ of the goods is a partition of $\items$ where agent $i$ receives the set of items $A_i$. We denote by $\allocations$ the set of all allocations of $\items$ (without disposal) to the set of agents $\agents$.
An instance of our setting is described by $\inst=(\agents,\items, \alpha,\allsignals,v_1,\ldots,v_n)$.

\paragraph{Fairness Notions.} We consider the following natural extensions of the well-studied fairness notions to the interdependent setting. For the case of equal entitlements (i.e., $\en{i}=\nicefrac{1}{n}$ for all $i\in\agents$), the \emph{envy}-based fairness notions that we consider are:
\begin{definition}[Envy-based Fairness Notions]\label{def:ef}
    An allocation $A$ is EF (respectively, EF1, EFX)  for agent $i \in \agents$ with respect to a signal vector $s=(s_1,\ldots,s_n)$ if, for all agents $i'\in \agents$ it holds that  
    \begin{align*}
        v_i(s,A_i) &\geq v_i(s,A_{i'}) \tag{EF}\\
        A_{i'}=\emptyset  \mbox{ or } \exists j \in A_{i'}:~  v_i(s,A_i) &\geq v_i(s,A_{i'}\setminus\{j\}) \tag{EF1}\\
        \forall j \in A_{i'}, ~v_i(s,A_i) &\geq v_i(s,A_{i'}\setminus\{j\}) \tag{EFX}
    \end{align*} 
An allocation $A$ is EF (respectively, EF1, EFX) with respect to signal vectors $s^{(1)},\ldots,s^{(n)}$ if for all $i\in\agents$, $A$ is EF (respectively EF1, EFX) for agent $i$ with respect to signal vector $s^{(i)}$.
If $s^{(1)}=\ldots=s^{(n)}=s$, we say that $A$ is EF (respectively, EF1 or EFX) with respect to $s$.
\end{definition}

The \emph{share}-based fairness notions that we consider are:
\begin{definition}[Share-based Fairness Notions]\label{def:share}
     The PROP  (respectivley MMS or APS) of agent $i$ with valuation $v_i:\allsignals \times 2^\items\rightarrow \R_{\geq 0 } $ with respect to signal vector $s=(s_1,\ldots,s_n)$  is:
    \begin{align*}
       &\prop{i}{\en{i}}{s} = \en{i} \cdot v_i(s, \items) \tag{PROP}\\
        &\mms{i
    }{\en{i}=\nicefrac{1}{n}}{s} =\max_{(A_1,\ldots,A_n)\in \allocations} \min_{i' \in \agents} v_i(s,A_{i'}) \tag{MMS}\\
        &\aps{i}{\en{i}}{s} =  \min_{p \in \prices} \max_{T\subseteq \items } v_i(s,T) \cdot \ind{\sum_{j\in T} p_j  \leq \en{i}} \tag{APS},
    \end{align*}
    where $\prices$ is the set of all non-negative price vectors that sum to $1$, and given a price vector $p\in \prices$, and a subset of items $T\subseteq \items$, we denote by $p(T)$ the sum of prices of items in $T$, (i.e., $p(T)=\sum_{j \in T} p_j$). 
    Note that the MMS is only defined for the case of equal entitlements.
    An allocation $A$ is PROP (respectively, MMS or APS) with respect to signal vectors $s^{(1)},\ldots,s^{(n)}$ if for every agent $i\in\agents$, the value of agent $i$ with respect to signal vector $s^{(i)}$ is at least the PROP (respectively, MMS or APS) of agent $i$ with respect to signal vector $s^{(i)}$.
    If $s^{(1)}=\ldots=s^{(n)}=s$, we say that $A$ is PROP (respectively, MMS or APS) with respect to $s$.
\end{definition}


\paragraph{Allocation Mechanisms.} 

\begin{definition}[Deterministic Mechanisms]\label{def:mech}
    A (deterministic) allocation mechanism $\mech$ (without payments) for the interdependent setting is defined by a product set of bids $B= \bigtimes_{i\in \agents }B_i$ and a mapping  $\mech:S \times B \rightarrow \allocations$. The allocation mechanism collects from each agent $i\in \agents $ a report $(r_i,b_i)\in S_i \times B_i$, and allocates the goods according  to $\mech(r,b)$, where $r=(r_1,\ldots,r_n)$, and $b=(b_1,\ldots,b_n)$. We also denote by $\mech_i(r,b)$ as the bundle of goods that agent $i$ receives under reports corresponding to $(r,b)$.
\end{definition}

\begin{definition}[Interdependent Pure Nash Equilibria]\label{def:equalibrium}    
For an instance $\inst$, and a mechanism $\mech$ for signal vector $s\in S$  a report vector $((r_1,b_1),\ldots,(r_n,b_n))$ is a pure Nash equilibrium (PNE) if for every agent $i \in \agents$ and report $(r_i',b_i')$ it holds that: $$ v_i(r^{(i)},\mech_i(r,b)) \geq v_i(r^{(i)},\mech_i(r',b')), $$
where $b=(b_1,\ldots,b_n),b'=(b_i',b_{-i}),r=(r_1,\ldots,r_n), r'=(r_i',r_{-i}), r^{(i)}=(s_i,r_{-i})$.
Note that the vector signal that agent $i$ calculates his value with respect to, is the reported signals of the other agents along with his own true signal.

\end{definition}


Some of our results apply to the special cases with additive, XOS, and subadditive valuations.
We say that an instance is additive if for every agent $i\in \agents$, every signal vector $s\in S$, and every subset $T\subseteq \items$ it holds that  $v_i(s,T)=\sum_{j \in T} v_i(s,\{j\})$. 
We say that an instance is XOS if for every agent $i\in \agents$, every signal vector $s\in S$, there exist an integer $\ell$ and $\ell$ non-negative vectors $a^1,\ldots,a^\ell$ of dimension $|\items|$, such that for every subset $T\subseteq \items$ it holds that  $v_i(s,T)=\max_{\ell' \in \{1,\ldots,\ell\}}\sum_{j \in T} a^{\ell'}_j$. 
We say that an instance is subadditive if for every agent $i\in \agents$, every signal vector $s\in S$, and every pair of subsets $T,T' \subseteq \items$ it holds that  $v_i(s,T) + v_i(s,T') \geq  v_i(s,T \cup T')$. 

\begin{remark}\label{remark:independent}
We note that independent private values can be captured by our model by setting all $v_i$ to be independent of $s_{-i}$. I.e., there exist functions $u_i: S_i \times 2^\items \rightarrow \R_{\geq 0 }$ such that for every agent $i\in\agents$, for every signal vector $s=(s_1,\ldots,s_n)\in \allsignals$, and every subset $T\subseteq \items$ it holds that $ v_i(s,T) = u_i(s_i,T) $.
\end{remark}

\section{The Case of 2 Agents with Equal Entitlements}\label{sec:equal}
In this section, we consider the case where there are two agents (i.e., $n=2$) with equal entitlements (i.e., $\en{1}=\en{2}=\nicefrac{1}{2}$) and general monotone valuation functions. We devise the IDV Cut-\&-Choose mechanism (Mechanism \ref{alg:cutchoose}), which involves agents reporting their own signal along with a guess of the other agent's signal. The mechanism then divides the goods into two sets in order to maximize the value of the set (according to the cutter's report) the chooser does not select (again, according to the cutter's report). The chooser then selects the better set based on his report. We show that the IDV Cut-\&-Choose induces an equilibrium where, once one agent declares her own signal truthfully and guesses correctly the other's signal, the other agent's best response is to also report truthfully his own signal, and repeat the signal of the other agent. The resulting allocation in this equilibrium is MMS and EFX for the cutter and EF for the chooser with respect to the true signals.

\subsection{The Interdependent Value Cut-\&-Choose Mechanism}
We consider both shared-based and envy-based notions of fairness and devise a mechanism that has guarantees of both types (EFX and MMS for the cutter, and EF for the chooser). Formally, we show that:
\begin{theorem}\label{thm:identical}
For every instance $\inst=(\agents,\items,\alpha=(\nicefrac{1}{2},\nicefrac{1}{2}), S=S_1\times S_2, v_1,v_2)$ composed of two agents with equal entitlements, there exists an allocation mechanism $\mech$ such that for every signal vector $s\in S$ there exists a report $((r_1,b_1),(r_2,b_2))$ such that: (1) Report $((r_1,b_1),(r_2,b_2))$ is a PNE; (2) The allocation $\mech((r_1,r_2),(b_1,b_2))$ is MMS and EFX for one agent with respect to the true signal vector $s$; (3) The allocation $\mech((r_1,r_2),(b_1,b_2))$ is EF for the other agent with respect to the true signal vector $s$.
\end{theorem}

To prove the above theorem, we devise the IDV Cut-\&-Choose mechanism (Mechanism \ref{alg:cutchoose}). Following \Cref{def:mech}, in Mechanism \ref{alg:cutchoose}, the set of bids of agents $1,2$, are the set of signals of agents $2,1$ respectively (i.e., $B_1=S_{2}$, and $B_2=S_1$). For ease of understanding, in the statements that follow, we subsume the notation of above.

In our proof and in the mechanism construction, we also use the following theorem that immediately derives from \citet[Theorem 4.2]{PlautR20}, for the case of independent valuations.
\begin{theorem}[\citep{PlautR20}]\label{thm:plautr20}
For every monotone valuation function $v:2^\items \rightarrow \R_{\geq 0}$, 
there exists a set $T^*$ such that $ v(T^*) \geq v(\items \setminus T^*) = MMS(v)$, and for every $j \in T^*$, it holds that $v(T^*\setminus \{j\}) \leq v(\items\setminus T^*)$, where $MMS(v)=\underset{T\subseteq \items}{\max}~\min(v(T),v(\items \setminus T))$. 
\end{theorem}

\begin{algorithm}
    \DontPrintSemicolon
    \SetAlgorithmName{Mechanism}{mechanism}{}
    \caption{IDV Cut-\&-Choose}\label{alg:cutchoose}
    \KwData{Report $(r_1, b_1)$ of agent 1 (cutter) and $(r_2,b_2)$ of agent 2 (chooser).}
    \Fn{\cutter}{
        Let $\cT(r_1,b_1) := \{T \subseteq M \mid v_2((r_1,b_1), T) \geq v_2((r_1,b_1), \items \setminus T)\}$\;
        Let $\xi_1(r_1,b_1) := \underset{T \in \cT(r_1,b_1)}{\max}~v_1((r_1,b_1), \items \setminus T)$\;
        \If{$\xi_1(r_1,b_1) > \mms{1}{\nicefrac{1}{2}}{(r_1, b_1)}$}{
            Let $T^*(r_1,b_1)$ be a set in $\cT(r_1,b_1)$ for which $v_1((r_1,b_1), \items \setminus T) =\xi_1(r_1,b_1)$\;
          }
          \Else{
            Let $T^*(r_1,b_1)$ be the set of \Cref{thm:plautr20}, when applied to $v_1((r_1,b_1),\cdot)$\;
          }
        Partition goods into subsets $(T^*(r_1,b_1), \items \setminus T^*(r_1,b_1))$\;
    }\;

    \Fn{\chooser}{
        Select $A_2 = \underset{T \in \{T^*(r_1,b_1), \items \setminus T^*(r_1,b_1)\}}{\arg\max} v_2((b_2,r_2), T)$, in case of a tie select $A_2 = T^*(r_1,b_1)$\;
    }
    \KwResult{Return allocation $\cM((r_1,r_2),(b_1,b_2)) = (A_1, A_2)$, where $A_1 = \items \setminus A_2$.}
\end{algorithm} 

We first observe that since for $T_{\text{MMS}}= \arg\max_{T\subseteq \items} \min(v_1(s,T),v_1(s,\items\setminus T))$, and since at least one of $T_{\text{MMS}},\items \setminus T_{\text{MMS}}$ is in $\cT(s_1,s_2)$, it holds that:  
\begin{align}\label{eq:xi}
    \xi_1(s_1,s_2)  & = \underset{T \in \cT(s_1,s_2)}{\max}~v_1((s_1,s_2), \items \setminus T) \nonumber\\ & \geq \min(v_1((s_1,s_2), T_{\text{MMS}}),v_1((s_1,s_2), \items \setminus T_{\text{MMS}})) = \mms{1}{1/2}{s}.
\end{align}

\begin{restatable}{lemma}{lemequilibrium}\label{lem:equilibrium}
    For Mechanism \ref{alg:cutchoose} and signal vector $s=(s_1,s_2) \in \allsignals$, report vector $ ((s_1,s_2),(s_2,s_1))$ is a PNE.
\end{restatable}
\begin{proof}
    Let us assume agent 1, the cutter, reports $(s_1, s_2)$. Then, since agent 2, the chooser, cannot influence the way goods are partitioned through his report, he has no incentive to deviate as if reporting $(s_2,s_1)$, he will be allocated the better of the two parts according to $(s_1,s_2)$. 
    
    On the other hand, suppose agent 2 reports $(s_2, s_1)$. In this case, if agent 1 reports $(s_1,s_2)$, then 
    $v_1((s_1,s_2),A_1) \geq \xi_1(s_1,s_2)$.
    This is since if we are in the \emph{if} branch of Mechanism~\ref{alg:cutchoose}, then $A_1=\items\setminus T^*(s_1,s_2)$, so it holds by definition of $T^*(s_1,s_2)$.     
    In the \emph{else} branch, we know that $\xi_1(s_1,s_2) = \mms{1}{\nicefrac{1}{2}}{s}$ by \Cref{eq:xi}. By \Cref{thm:plautr20}, both bundles $T^*(s_1,s_2)$ and $\items \setminus T^*(s_1,s_2)$ have values of at least $\mms{1}{\nicefrac{1}{2}}{s}$ (i.e., $
    v_1((s_1,s_2),T^*(s_1,s_2)) ,v_1((s_1,s_2), \items \setminus T^*(s_1,s_2)) \geq \mms{1}{\nicefrac{1}{2}}{s}$), which means that agent 1 gets again at least $\xi_1(s_1,s_2)$ by reporting $(s_1, s_2)$.
    
    If agent 1 deviates and reports some $(r_1',b_1') \neq (s_1,s_2)$, then the split induced by this declaration will be composed of two sets $(A_1', A_2')$. Let $\widetilde{T}$ be the subset agent 2 selects, i.e., $v_2((s_1, s_2), \widetilde{T}) \geq v_2((s_1, s_2), \items \setminus \widetilde{T})$. This implies that $\widetilde{T} \in \cT(s_1,s_2)$. However, assume towards contradiction that the deviation was beneficial, i.e., $v_1((s_1, s_2), \items \setminus \widetilde{T}) > \xi_1(s_1,s_2)$. Then 
    $$\xi_1(s_1,s_2) = \max_{T \in \cT(s_1,s_2)} v_1((s_1,s_2), \items \setminus T)\geq  v_1((s_1,s_2), \items \setminus \widetilde{T}) >\xi_1(s_1,s_2),$$ which is a contradiction, and agent 1 cannot benefit from a deviation.
\end{proof}

\begin{restatable}{lemma}{lemcutter}\label{lem:cutter}
    For Mechanism \ref{alg:cutchoose} and signal vector $s=(s_1,s_2)\in \allsignals$, in PNE $((s_1,s_2),(s_2,s_1))$, the cutter receives at least her MMS with respect to signal vector $s$. Moreover, the allocation $\mech((s_1,s_2),(s_2,s_1))$ is EFX for the cutter with respect to signal vector $s$.
\end{restatable}
\begin{proof}
    In the equilibrium $((s_1, s_2), (s_2, s_1))$, agent 1 gets $A_1$. 
    In the proof of \Cref{lem:equilibrium} we proved that $v_1((s_1,s_2),A_1) \geq \xi_1(s_1,s_2)$. Combining this with \Cref{eq:xi} shows that the cutter gets at least her MMS with respect to signal vector $s$.
    
    To show that the obtained allocation is EFX for the cutter with respect to $s$, we distinguish into two cases, following how the mechanism proceeds. In the first case (the \emph{if} branch of Mechanism \ref{alg:cutchoose}), $\xi_1(s_1, s_2) > \mms{1}{\nicefrac{1}{2}}{s}$ and we know that $v_1(s, A_1) \geq  \xi_1(s_1, s_2)$. Also, by definition of MMS, at least one subset in each partition into two subsets, must have a value of at most the MMS. Thus, $v_1(s, A_2) \leq \mms{1}{\nicefrac{1}{2}}{s} < v_1(s,A_1)$, so that the cutter does not envy the chooser. 
    
    Otherwise (the \emph{else} branch of Mechanism \ref{alg:cutchoose}), by \Cref{eq:xi}, $\xi_1(s_1, s_2) = \mms{1}{\nicefrac{1}{2}}{s}$, in which case we know, by \Cref{thm:plautr20}, no matter which set agent 1 receives from partition $(A_1,A_2)$, the allocation will be EFX for her with respect to signal vector $s$.
\end{proof}

Since the chooser selects the bundle with higher value according to his report, we can observe the following:
\begin{observation}\label{obs:ef}
     For Mechanism \ref{alg:cutchoose} and signal vector $s=(s_1,s_2)\in \allsignals$, in equilibrium $((s_1,s_2),(s_2,s_1))$, the allocation $\mech((s_1,s_2),(s_2,s_1))$ is EF for the chooser  with respect to $s$.
\end{observation}

The proof of \Cref{thm:identical} follows by combining \Cref{lem:equilibrium,lem:cutter,obs:ef}.

\subsection{Discussion on Mechanism~\ref{alg:cutchoose}}

We next discuss the extent of \Cref{thm:identical}, and some of its implications. Regarding envy-based notions of fairness, it is clear that there might not exist an allocation that is EF for both agents (e.g., two agents with a single good). Thus, guaranteeing EF for one agent and EFX for the other is the best we can aim for.

For what concerns shared-based notions of fairness, giving the MMS to both agents is impossible for general (subadditive or even XOS) valuations as the next example illustrates.\footnote{We omit the dependence on $s$, since this also holds in non-interdependent settings (and even without incentive considerations).} Consider set of goods $\items = \{a, b, c, d\}$, valuation $v_1 = \max\left(\ind{a \vee b} +\ind{a \wedge b}, \ind{c \vee d} +\ind{c \wedge d} \right)$ for agent 1, and $v_2 = \max\left(\ind{a \vee d} +\ind{a \wedge d}, \ind{b \vee c} +\ind{b \wedge c} \right)$ for agent 2. We observe that (1) these valuations are XOS, and (2) the MMS  for both agents is $2$. However, no allocation gives the MMS  to both agents: If one agent receives one good (and the other three goods), her value is at most $1$. Otherwise, both agents receive two goods, and one of them must have a value of at most $1$ which is strictly less than the MMS. 

In case of subadditive valuations, it holds that EF implies PROP. Therefore, if the chooser's valuation  is subadditive then in equilibrium $((s_1,s_2),(s_2,s_1))$, Mechanism \ref{alg:cutchoose} guarantees agent 2 his proportional share with respect to signal vector $s$, 
and when the chooser's valuation is additive, then it guarantees the MMS with respect to signal vector $s$ to both agents. 

In the special case of identical valuations (i.e., $v_1=v_2$), the MMS can be guaranteed to both agents (for all valuation classes). Indeed, when cutter and chooser have identical valuations, the chooser receives at least his MMS with respect to signal vector $s$, in equilibrium $((s_1,s_2),(s_2,s_1))$ of Mechanism \ref{alg:cutchoose}. This is because the chooser gets a value of at least $\xi_1(s_1,s_2)$ which is at least the MMS.
\section{The Case of 2 Agents with Unequal Entitlements}\label{sec:unequal}
This section considers the case where two agents (i.e., $n=2$) have unequal entitlements and general monotone valuation functions. 
Our main contribution is the design of the IDV Price-\&-Choose mechanism (Mechanism \ref{alg:pricechoose}).
Mechanism \ref{alg:pricechoose}  first prices the goods according to the pricer's report. The pricing tries to maximize the pricer's value of the remaining goods (according to the pricer's report) while assuming that the chooser first picks the set that maximizes his own value (again, according to the pricer's report), and whose total price does not exceed the chooser's entitlement.
As for the equal entitlement case, this mechanism guarantees that there exists a PNE of the same form, of reporting your own signal truthfully and guessing the other's signal correctly. The resulting allocation gives at least the PROP share to the chooser as long as his valuation is XOS, and the APS to the pricer. We, in fact, show that achieving the same fairness guarantees is impossible if both valuation functions are subadditive.

\subsection{The Interdependent Value Price-\&-Choose Mechanism}
We show the following result:

\begin{theorem}\label{thm:different}
For every instance $\inst=(\agents,\items,\alpha=(\en{1},\en{2}), S=S_1\times S_2, v_1,v_2)$ composed of two agents with unequal entitlements with monotone valuations, there exists an allocation mechanism $\mech$ such that for every signal vector $s\in S$ there exists a report $((r_1,b_1),(r_2,b_2))$ such that:
  (1) Report $((r_1,b_1),(r_2,b_2))$ is a PNE;
   (2) The allocation  $\mech((r_1,r_2),(b_1,b_2))$ is APS for one agent with respect to the true signal vector $s$; 
    (3) If the other agent's valuation is XOS, then the allocation  $\mech((r_1,r_2),(b_1,b_2))$ is PROP for the other agent with respect to the true signal vector $s$.
\end{theorem}

We devise the IDV Price-\&-Choose mechanism (Mechanism \ref{alg:pricechoose}) to prove the above theorem. Following \Cref{def:mech}, in Mechanism \ref{alg:pricechoose}, the set of bids of agents $1,2$, are the set of signals of agents $2,1$ respectively (i.e., $B_1=S_{2}$, and $B_2=S_1$).

\begin{algorithm}
    \DontPrintSemicolon
    \SetAlgorithmName{Mechanism}{mechanism}{}
    \caption{IDV Price-\&-Choose\label{alg:pricechoose}}
    \KwData{Report $(r_1, b_1)$ of agent 1 (pricer) and $(r_2,b_2)$ of agent 2 (chooser).}
    \Fn{\pricer}{
        Let $\xi_2(r_1,b_1, p) := \underset{T \subseteq M }{\max}~v_2((r_1,b_1), T) \cdot \ind{p(T) \leq \en{2}}$\;
        Let $\cT(r_1,b_1, p) := \{T \subseteq M \mid  p(T) \leq \en{2} \wedge  v_2((r_1,b_1), T) = \xi_2(r_1,b_1, p)$\}\;   
        Let $\cT(r_1,b_1) := \bigcup_{p \in \prices}\cT(r_1,b_1, p) $\;
        Let $T^*(r_1,b_1) := \underset{T \in \cT(r_1,b_1)}{\arg\max}~v_1((r_1,b_1),\items\setminus T)$\;
        Let $p^* \in \prices$ be a price vector such that $T^*(r_1,b_1) \in \cT(r_1,b_1, p^*)$\;
        Offer items for prices $p^*$\;
    }
    \Fn{\chooser}{
        \If{$v_2((b_2,r_2),T^*(r_1,b_1)) = \underset{T \subseteq \items: p^*(T) \leq \en{2}}{\max} v_2((b_2,r_2), T)$}{
        Select $A_2 = T^*(r_1,b_1)$\;
        }
        \Else{
            Select an arbitrary set $A_2$ in $\underset{T \subseteq \items: p^*(T) \leq \en{2}}{\arg\max} v_2((b_2,r_2), T)$\;
        }
    }
    \KwResult{Return allocation $\cM((r_1,r_2),(b_1,b_2)) = (A_1, A_2)$, where $A_1 = \items \setminus A_2$.}
\end{algorithm}

\begin{restatable}{lemma}{lemequilibriumaps}\label{lem:equilibrium-aps}
    For Mechanism \ref{alg:pricechoose} and signal vector $s=(s_1,s_2) \in \allsignals$, report vector $ ((s_1,s_2),(s_2,s_1))$ is a PNE.
\end{restatable}
\begin{proof}
    Let us assume agent 1, the pricer, reports $(s_1, s_2)$. Then, since agent 2, the chooser, cannot influence the way goods are priced through his report, he has no incentive to deviate as if reporting $(s_2,s_1)$, he will choose the subset $A_2$ that maximizes his value according to $(s_1,s_2)$, subject to his entitlement constraint (i.e., $p^*(A_2) \leq \en{2}$). 
    
    On the other hand, suppose agent 2 reports $(s_2, s_1)$. In this case, if agent 1 reports some $(r_1',b_1') \neq (s_1,s_2)$, this induces a price vector $\widetilde{p}$ such that agent 2 selects a set $\widetilde{T} \in \cT(s_1,s_2, \widetilde{p})$, and therefore in $\widetilde{T} \in \cT(s_1,s_2)$. Thus, by definition of $T^*(s_1,s_2)$, it holds that  $v_1((s_1,s_2),\items \setminus T^*(s_1,s_2)) \geq v_1((s_1,s_2),\items\setminus \widetilde{T})$, which concludes the proof.
\end{proof}

\begin{restatable}{lemma}{lemaps}\label{lem:aps}
    For Mechanism \ref{alg:pricechoose} and signal vector $s=(s_1,s_2)\in \allsignals$, in PNE $((s_1,s_2),(s_2,s_1))$, if the pricer's valuation is XOS, then she receives at least her PROP with respect to $s$.
\end{restatable}
\begin{proof}
    Since $v_1$ is XOS, then there exist an integer $\ell$ and $\ell$ non-negative vectors $a^1,\ldots,a^\ell$ of dimension $|\items|$, such that for every subset $T\subseteq \items$ it holds that  $v_i(s,T)=\max_{\ell' \in \{1,\ldots,\ell\}}\sum_{j \in T} a^{\ell'}_j$.
    Let $\ell^*$ be an arbitrary index maximizing $\sum_{j\in \items} a^{\ell^*}_j$, and let $\delta = \sum_{j\in \items} a^{\ell^*}_j$.
    
    If $\delta=0$, then PROP of agent $1$ with respect to signal $s$ is $0$, thus any allocation is PROP for agent $1$ with respect to signal $s$.
    Else, consider the pricing $\widetilde{p}$, where item $j\in\items$ is priced $\widetilde{p}_j=\nicefrac{a^{\ell^*}_j}{\delta}$. Let $\widetilde{T}$ be an arbitrary set in  $\cT(s_1,s_2, \widetilde{p})$. It holds that  
    \begin{align*}
        v_1(s,\items \setminus \widetilde{T}) &\geq \sum_{j \in \items \setminus \widetilde{T}} a^{\ell^*}_j = \delta -\sum_{j \in  \widetilde{T}} a^{\ell^*}_j  = \delta \cdot \left(1-\sum_{j \in  \widetilde{T}} \widetilde{p}_j\right) \\
        &\geq \delta \cdot \left(1-\en{2}\right) = \delta \cdot \en{1} =\prop{1}{\en{1}}{s},
    \end{align*}
    where the first inequality is by the definition of XOS function, and the second inequality is since  $\sum_{j \in \widetilde{T}}  \widetilde{p}_j \leq \en{2}$.
    This concludes the proof.    
\end{proof}

Since the chooser selects the bundle with the highest value according to his report subject to his entitlement, we can observe the following:
\begin{observation}\label{obs:aps}
     For Mechanism \ref{alg:pricechoose} and signal vector $s=(s_1,s_2)\in \allsignals$, in equilibrium $((s_1,s_2),(s_2,s_1))$, the allocation $\mech((s_1,s_2),(s_2,s_1))$ is APS for the chooser with respect to $s$.
\end{observation}

The proof of \Cref{thm:different} follows by combining \Cref{lem:equilibrium-aps,lem:aps,obs:aps}.

\subsection{Discussion on Mechanism~\ref{alg:pricechoose}}
We showed that Mechanism~\ref{alg:pricechoose} under XOS valuations gives in equilibrium $((s_1,s_2),(s_2,s_1))$ the pricer her PROP, and the chooser the APS with respect to the true signals.
This implies that in the independent values model, for two XOS agents there is always an allocation that gives one agent her proportional share, and the other his APS.
In the special case where the pricer's valuation is additive, since PROP is at least as large as the APS \citep{BabaioffEF21}, this guarantees the APS for both agents. Moreover, for our Price-\&-Choose mechanism when applied to independent additive valuation functions, all equilibria are APS fair with respect to the true values of the agents. Thus, we have the following corollary:
\begin{corollary}\label{cor:aps-independent}
    When applied to the independent values model  (see \Cref{remark:independent}) for additive agents, Mechanism~\ref{alg:pricechoose} guarantees that in every equilibrium, both agents receive their APS.
\end{corollary}
Giving both agents their proportional share is trivially impossible even for the simple case of additive valuations (e.g., when there is a single good).
We next show that there are subadditive instances with two agents (even in the non-interdependent case, with no incentive constraints) for which there is no allocation that gives one agent her APS and the other her PROP. This is true even for agents with equal entitlements and same valuation.
\begin{restatable}{proposition}{propsubadd}
    There exists a subadditive valuation $v:2^\items\rightarrow \R_{\geq 0}$, such that for every set $T\subseteq \items$, either $v(T) < \nicefrac{v(\items)}{2} =\text{PROP}$ or $v(\items \setminus T) <\min_{p\in \prices} \max_{T'\subseteq \items } v(T') \cdot \ind{\sum_{j\in T'}p_j \leq \nicefrac{1}{2}}=\text{APS}$. 
\end{restatable}
\begin{proof}
    Our construction of $v$ is based on a similar construction from \citep[Section 3]{EzraFNTW19}.
    Consider the set of items $\items$ of all non-zero binary vectors of size $k$ for some even $k\geq 6$ (i.e., $\items=\{0,1\}^k \setminus \{(0,\ldots,0)\}$).
    For a vector $u\in \{0,1\}^k$, we denote by $B_u=\{j \mid \langle j,u \rangle =0\}$ (and the inner product is defined over $\F_2$).
    Let $g(T):2^\items \rightarrow \R_{\geq 0}$ be a set cover function where the collection of sets is $\{B_u \mid u\in \{0,1\}^k\}$, i.e., $$g(T) =\min\left\{ \ell  \mid \exists u_1,\ldots ,u_\ell\in\{0,1\}^k: T \subseteq \bigcup_{\ell'=1}^\ell B_{\ell'}\right\} .$$
    Note that $g$ is monotone and subadditive since it is a set cover function.
    We also observe that $g(\items) = k$.
    This is true since  for every collection of vectors $u_1,\ldots,u_{k-1}\in \items$, there exists a vector $u$ such that $\langle u,u_\ell \rangle =1$, for every $\ell=1,\ldots, k-1$, and therefore $u\notin \bigcup_{\ell=1}^{k-1} B_{u_\ell} $, which implies that $\bigcup_{\ell=1}^{k-1} B_{\ell} \neq \items$. 
    We now define the following function 
    \begin{equation*}
        v(T) = \begin{cases}
            g(T) \quad &\text{ if } g(T) < \frac{k}{2} \\
            k-g(\items\setminus T) \quad &\text{ if } g(\items \setminus T) < \frac{k}{2} \\
           k/2 \quad &\text{ else}
      \end{cases}
    \end{equation*}
    \citet[Lemma 3.2, Lemma 3.3]{EzraFNTW19} show that for a set cover function $g$ such that $g(\items)\geq k$, the function $v$ possesses the following properties: (1) $v$ is well-defined (it cannot occur that both $g(T), g(M \setminus T) < \nicefrac{k}{2}$ as otherwise there would exist a collection of $k-1$ sets $B_u$'s that cover the entire goods set $\items$), and hence the cases that define $v$ are disjoint; (2) $v$ is monotone and subadditive; (3) for all $T\subseteq \items $, $v(T) + v(\items \setminus T) = k$.

    We next show that the APS of $v$ with entitlement $\nicefrac{1}{2}$ is at least $k-2$, that is $$\min_{p\in \prices} \max_{T'\subseteq \items } v(T') \cdot \ind{\sum_{j\in T'}p_j \leq \frac{1}{2}} \geq k-2.$$ Consider an arbitrary prices $p\in \prices$. To this end, assume towards contradiction that for all non-zero vectors $u \in \items $, we have $p(\items \setminus B_u) > \nicefrac{|\items \setminus B_u|}{|\items|}$. Then, $\sum_{u \in \items} p(\items \setminus B_u) > |\items \setminus B_u| = \frac{m+1}{2}$. On the other hand, we also have that
    \begin{align*}
        \sum_{u \in \items} p(\items \setminus B_u) &= \sum_{u \in \items}\sum_{j \in \items} p(\{j\}) \cdot \ind{j \in \items \setminus B_u} \\
        &= \sum_{j \in \items} p(\{j\}) \cdot \sum_{u \in \items} \ind{j \in \items \setminus B_u} = \sum_{j \in \items} p(\{j\}) \cdot \frac{m+1}{2} = \frac{m+1}{2},
    \end{align*}
    where the second to last equality follows from the fact that any good $j$ appears in exactly $\nicefrac{(m-1)}{2}$ $B_u$'s (excluding the case where the inner product is 0) and the last equality is since $\sum_{j \in \items} p(\{j\}) = 1$. A contradiction. 
    
    Thus, there exists a good $u^*$ such that $p(\items \setminus B_{u^*}) \leq\nicefrac{|\items \setminus B_{u^*}|}{m} =\nicefrac{(m+1)}{2m} $. So there must exist a good $j \in \items \setminus B_{u^*}$ such that $p(\items \setminus (B_{u^*} \cup \{j\})) \leq \nicefrac{1}{2}$. Hence, $v(\items \setminus (B_{u^*} \cup \{j\})) = k-2$, where the last equality holds since $g(B_{u^*} \cup \{j\}) = 2$ (covering $B_{u^*}$ costs 1 and covering singleton $\{j\}$ also costs 1). Thus, for every price vector $p$, there exists a bundle with a price of at most $\nicefrac{1}{2}$ with a value of $k-2$, and therefore the APS is at least $k-2$.

    Since $v(\items) =k $, then PROP is $\nicefrac{k}{2}$. The proposition then holds since for every set $T\subseteq \items$, $$v(T)+v(\items\setminus T) = k <   \nicefrac{k}{2}+k-2\leq \text{PROP} + \text{APS},$$ which finishes the proof.    
\end{proof}
\section{The Case of 3 or More Agents}\label{sec:multi}

In this section we show how for the case of $n\geq 3$ agents, we can reduce the problem of designing a mechanism that has at least one fair equilibrium with respect to the true signal $s$, to the purely algorithmic problem. On a high level, we do so as follows: Let each agent report her signal as well as everybody else's. Since $n \geq 3$, if at least $n-1$ (unanimity up to one vote) agents agree on their reports and guesses, we can execute the algorithm with respect to the reported declarations, excluding the agent that did not agree (if any). This implies that when all the agents bid the same, the produced allocation is identical to the one where all the agents report as before, and just one agent deviates to something different\footnote{It is easy to see that this technique cannot be implemented when there are only two agents, which makes it more challenging.}. The latter is crucial to guarantee a PNE with the algorithm's properties.

To formalize this reduction, consider any algorithm $\mathcal{F}:V_1\times,\ldots,V_n \rightarrow \allocations$ that receives for every agent $i$ an (independent) valuation $v_i:2^\items \rightarrow \R_{\geq 0} $ from a set of valuations $ V_i$, and returns an allocation that satisfies some property $X$ (where property $X$ can be EF, EFX, EF1, PROP, MMS, APS, and even non-fairness properties such as maximizing social welfare, or Pareto optimality). Then we prove the following:

\begin{restatable}{theorem}{thmblackbox}\label{thm:blackbox}
For every property $X$, every algorithm $\mathcal{F}:V_1\times \ldots \times V_n \rightarrow \allocations$ that always satisfies property X, and every instance $\inst=(\agents,\items,\alpha,\allsignals,v_1,\ldots,v_n)$ with $n \geq 3$ agents, for which for all $i$, and every signal vector  $s\in\allsignals$, it holds that $v_i(s,\cdot) \in V_i$,
there exists an allocation mechanism $\mech$ such that for every signal vector $s\in \allsignals$ there is at least one PNE,   $((r_1,b_1),\ldots,(r_n,b_n))$, that satisfies that the allocation $\mech((r_1,\ldots,r_n),(b_1,\ldots,b_n))$ satisfies property $X$ with respect to $s$. 
\end{restatable}
\begin{proof}
To prove the theorem, we devise the following mechanism, for which the set of bids of agent $i$, is the product of the set of signals of all agents, i.e., $B_i=\allsignals $, and $B = B_1\times \ldots \times B_n$.
We denote by $b_{i,j}$ the $j^{\text{th}}$ coordinate of vector $b_i$.
\begin{algorithm}
    \DontPrintSemicolon
    \SetAlgorithmName{Mechanism}{mechanism}{}
    \caption{IDV Black Box Transformation from $\mathcal{F}$ to $\cM$}\label{alg:blackbox}
    \KwData{Report vector $((r_1,b_1), (r_2,b_2), \ldots (r_n,b_n))$.}
        \If{ \normalfont there exists  a set $\agents' \subseteq \agents$ of size at least $n-1$, such that 1) for every $i,i'\in \agents'$, we have $ b_i=b_{i'}$, and 2) for every $i\in \agents'$ it holds that $r_i=b_{i,i}$ }{
        Select an arbitrary $i^*\in \agents'$
           Set $\cM(r,b)=\mathcal{F}(v_1(b_{i^*},\cdot),\ldots,v_n(b_{i^*},\cdot))$\;
          }
          \Else{Set $\cM(r,b)=A=(A_1, A_2,\ldots, A_n)$, where $A$ is some predefined allocation\;
          }
    \KwResult{Return allocation $\cM(r,b)$}
\end{algorithm} 

The mechanism is well defined since because $n\geq 3$, there cannot be more than one value of $b_{i^*}$ that has $n-1$ agents that report it.

Consider the report vector $(s_1,s),\ldots,(s_n,s)$. It is clear that since more than $n-1$ agents agree on the value of $b_i$ (since all of them are $s$), and for every agent $i$,  $r_i=s_i=b_{i,i}$, then the output of Mechanism \ref{alg:blackbox} is $\mathcal{F}(v_1(s,\cdot),\ldots,v_n(s,\cdot))$, which satisfies all properties that $\mathcal{F}$ satisfies with respect to signal vector $s$. 
    Therefore, the only thing that remains to show, is that the report vector $(s_1,s),\ldots,(s_n,s)$ is a PNE. This follows directly from the construction of Mechanism \ref{alg:blackbox}, as in the case that one agent deviates, the produced allocation remains the same.
\end{proof}

\section{Impossibilities}\label{sec:main_imp}

So far, we considered allocation mechanisms that for every signal vector $s\in \allsignals$, have at least one PNE, and the allocation that corresponds to one of these PNE is fair (according to some fairness criterion) with respect to the true signals $s$. 
The next natural direction is to explore whether there are mechanisms that always admit a PNE for every signal vector $s$, and \textit{all} equilibria are fair with respect to the true signal vector $s$. 

In the following theorem, we consider the notions of MMS and EF1, and we show that this is impossible. Notice that this creates a separation with the strategic model of non-interdependent values, where it is known that there are allocation mechanisms with this property under agents with additive valuation functions, both for the case of MMS (2 agents), and the case of EF1 ($n$ agents). 

\begin{restatable}{theorem}{thmcounter}\label{thm:counter}
    There is no IDV-allocation mechanism that for every signal vector $s\in \allsignals$ has a PNE, and all of its equilibria induce MMS (or EF1) allocations with respect to the true signal vector $s$.
Moreover, this holds even for the special case where all valuation functions are additive dichotomous (where each item has a value of either $0$ or $1$).
\end{restatable}
\begin{proof}
We begin with the notion of MMS. Consider an instance $\inst$ with $n$ agents and a set of goods $\items $ of size $m=n^2$,  equal entitlements (i.e., $\en{1}=\ldots=\en{n}=\frac{1}{n}$). The set of potential signals of agent $1$ is $\{0,1\}^{m}$, and the signal space of all other agents is a singleton $\{1\}$. The binary additive valuation function of all the agents is
$$ v_i(s,T) = \sum_{j\in T} s_{1j},$$
where $s_{1j}$ is the $j^{\text{th}}$ coordinate of the signal $s_1$. 
Now suppose that for every $j \in \items$, we have that $s_{1j}=1$. This implies that $v_i(s,\{j\})=1$, for every $i \in N$ and $j \in M$. Since there are $n^2$ goods in total, it is easy to verify that the MMS of every agent, with respect to the true signal vector $s$, is $n$. Therefore, an allocation that guarantees the MMS to all the agents, according to $s$, should allocate to every agent exactly $n$ goods.

Now consider a mechanism $\mech$, and assume towards contradiction that for every signal vector $s\in \allsignals$ it has at least one PNE, and all of its PNE are MMS-fair with respect to the true signal vector $s$. This implies that there is a report vector $(r_1,b_1),\ldots,(r_n,b_n)$, that is a PNE, and since this PNE has to be MMS-fair with respect to the true signals of the agents $s$, it induces some allocation $A=(A_1, A_2, \dots, A_n)$ such that every agent gets exactly $n$ goods (i.e., $|A_1|=\ldots=|A_n|=n$).

Now consider a different signal vector $s'$, such that 
\[ 
		s'_{1j}= \left\{
		\begin{array}{ll}
			1 & j \in A_1\\
			0 & j \in M\setminus A_1. \\
		\end{array} 
		\right. 
		\]
The crucial observation now is that since the valuation function of every agent $i \in\{2,3,\ldots n\}$ depends only on the report of agent 1, no agent $i \in\{2,3,\ldots n\}$ has a profitable deviation under report vector $(r_1,b_1),\ldots,(r_n,b_n)$, as the perceived values that they have over the goods, are the same between them, no matter what the true signal of agent 1 is. At the same time, agent 1 also has no profitable deviation as set $A_1$ is the only set of goods that has value for her. Therefore, report vector $(r_1,b_1),\ldots,(r_n,b_n)$, is a PNE. Now, according to the true signal vector $s'$, for every agent $i$ we have that 
\[ 
		v_i(s',\{j\})= \left\{
		\begin{array}{ll}
			1 & j \in A_1 \\
			0 &j \in M\setminus A_1. \\
		\end{array} 
		\right. 
		\]
It is easy to verify that the MMS value of every agent $i \in N$ is 1, while for every $i \in \{2,3,\ldots,n\}$, we have $v_i(s', A_i)=0$. Therefore, report vector $(r_1,b_1),\ldots,(r_n,b_n)$ is a PNE that corresponds to an allocation that it not MMS with respect to the true signal vector $s'$, a contradiction.  

The example above also shows an impossibility result with respect to the notion of EF1.
It is possible to strengthen the same claim with respect to EF1 by considering  an instance with  $m=2\cdot n$ goods (instead of $n^2$). Again, only agent $1$ influences the values, and for every agent $i\in\agents$ and good $j\in \items$, $v_{ij}=s_{1j}$. Under signal vector $s$ (where for good $j \in M$, we have that $s_{1j}=1$), in an EF1 allocation $A=(A_1, A_2, \ldots, A_n)$ each agent should have exactly 2 goods (if someone gets 3 or more goods, then there will be an agent that gets at most 1 good, which contradicts the EF1 property with respect to $s$). Thus, there exists a report vector $(r_1,b_1),\ldots,(r_n,b_n)$, that is a PNE which corresponds to an EF1 allocation $A$ under the true signal vector $s$. This report vector is also a PNE under the signal vector $s'$ (where $s_{1j}=1$, if $j \in A_1$, and $0$ otherwise). The corresponding allocation of this report is not EF1  with respect to true signal vector $s'$, a contradiction.
\end{proof}

\begin{remark}
   \normalfont Note that our counter-example considers additive valuation functions, with equal entitlements, and binary values, for which in the independent settings, there is a truthful allocation mechanism that is MMS and EF1 fair \citep{Aziz20,HalpernPPS20}. This creates a separation between the two models.
\end{remark}
\section{Future Directions}

In this paper, we initiated the study of fair division of indivisible goods, under agents with interdependent values. We showed that despite the complexity of this setting, the design of mechanisms that have at least one pure Nash equilibrium, which is fair with respect to the true signals of the agents, is still possible. In addition, the mechanisms that we designed, provide a variety of fairness guarantees, even for agents with valuation functions that go beyond the additive class. Finally, we also presented a negative result showing that it is not possible to have mechanisms that always admit PNE, and all the PNE of which, are fair (with respect to either MMS or EF1). The latter creates a clear separation between the interdependent and the independent valuation function setting.

Our work leaves several interesting questions open. In particular, for the case of 2 agents, it would be compelling to explore whether it is possible to design mechanisms that, besides having at least one fair PNE for every instance, also have additional \emph{efficiency} properties, e.g., Pareto Optimality (as this is not the case with our mechanisms). Moreover, moving to the general case of $n$ agents, the fair PNE that our mechanisms have, are of a very specific form, where an agent has to more or less \emph{guess} the true signals of the other agents. As this may be practically unappealing, it would be nice to examine if we can have mechanisms with fair equilibria of simpler structure. Finally, another interesting direction would be to see whether we could identify meaningful settings of interdependent values, where the impossibility results that we present do not apply, i.e., settings where it is possible to design allocation mechanisms, for which all the PNE are fair.

\section*{Acknowledgments}
This project is supported by the ERC Advanced Grant 788893 AMDROMA, EC H2020RIA project “SoBigData++” (871042), PNRR MUR project PE0000013-FAIR”, PNRR MUR project  IR0000013-SoBigData.it.

\bibliography{references}

\end{document}